\documentclass{amsart}
\usepackage{color}
\usepackage{setspace}
\usepackage{wrapfig}
\usepackage{lineno}
\usepackage{graphicx}
\usepackage{amsmath}
\usepackage{amsfonts}
\usepackage[sort&compress,numbers]{natbib}
 \usepackage[colorlinks=true]{hyperref}
 \hypersetup{
    colorlinks=true,
    linkcolor=blue,
    filecolor=magenta,
    urlcolor=cyan,
    citecolor=blue,
}
\usepackage{url}
\usepackage{amsmath,amsxtra,amssymb,latexsym,epsfig,amscd,amsthm,fancybox,epsfig}
\usepackage[mathscr]{eucal}
\usepackage{graphicx}
\usepackage{multicol,xcolor}
\usepackage{epsfig} 
\usepackage{epstopdf}
\usepackage{cases}
\usepackage{subfig}
\usepackage{color}
\usepackage{hyperref}
\usepackage{bigints}

\newtheorem{thm}{Theorem}[section]

\newtheorem{prop}{Proposition}

\theoremstyle{definition}

\theoremstyle{remark}
\newtheorem{rem}{Remark}

\makeatletter
\def\munderbar#1{\underline{\sbox\tw@{$#1$}\dp\tw@\z@\box\tw@}}
\makeatother
\newcommand{\s}{\mathcal{S}}

\newcommand{\E}{\mathbb{E}}
\renewcommand{\P}{\mathbb{P}}
\newcommand{\R}{\mathbb{R}}

\begin{document}
\title[The $P^*$ rule in the stochastic Holt-Lawton model]{The $P^*$ rule in the stochastic Holt-Lawton model \\of apparent competition}
\author[S.J. Schreiber]{}
\address{Department of Evolution and Ecology, University of California, Davis, CA 95616 USA}
\email{sschreiber@ucdavis.edu}


\thanks{The author is supported by U.S. National Science Foundation grant DMS-1716803}

\maketitle

\centerline{\scshape Sebastian J. Schreiber$^*$}
\medskip
{\footnotesize
 \centerline{Department of Evolution and Ecology}
 \centerline{ and Center for Population Biology}
   \centerline{University of California}
   \centerline{Davis, CA 95616, USA}
} 

\begin{abstract}
In $1993$, Holt and Lawton introduced a stochastic model of two host species parasitized by a common parasitoid species. We introduce and analyze a generalization of these  stochastic difference equations with any number of host species, stochastically varying parasitism rates, stochastically varying host intrinsic fitnesses, and stochastic immigration of parasitoids. Despite the lack of direct, host density-dependence, we show that this system is dissipative i.e. enters a compact set in finite time for all initial conditions. When there is a single host species, stochastic persistence and extinction of the host is characterized using external Lyapunov exponents corresponding to the average per-capita growth rates of the host when rare. When a single host persists, say species $i$, a explicit expression is derived for the average density, $P_i^*$, of the parasitoid at the stationary distributions supporting both species. When there are multiple host species, we prove that the host species with the largest $P_i^*$ value stochastically persists, while the other host species are asymptotically driven to extinction. A review of the main mathematical methods used to prove the results and future challenges are given.  
\end{abstract}

\section{Introduction}

\citet{volterra-26} proved that competition for a single, limiting resource results in competitive exclusion via the $R^*$ rule: the competing species that suppresses the resource to the lowest equilibrium density excludes the other competing species~\citep{tilman-82}. Volterra's mathematical derivation was for ordinary differential equation models where the per-capita growth rates of the competing species are linear functions of resource availability~\citep[see discussion in][]{hofbauer_sigmund1998}. Since this work of Volterra, MathSciNet lists $279$ publications on the ``competitive exclusion principle'' of which $19$ appeared in \emph{Discrete and Continuous Dynamical Systems: Series B}~\cite{MR3327897,MR2912073,MR3986255,MR3986226,MR3693844,MR3639164,MR3639160,MR3639157,MR3432309,MR3327906,MR3327904,MR3228878,MR2525150,MR2300323,MR2291867,MR2224875,MR2129373,kuang_nagy2004,MR1821409}. These $19$ papers proved new principles of competitive exclusion for a diversity of situations including spatial chemostat models~\cite{MR3327897}, within-host competition of multiple viral types~\cite{MR3327906}, competing technologies~\cite{MR3327904}, epidemiological models of competing disease strains~\cite{MR2129373}, stoichiometric models of tumor growth~\cite{kuang_nagy2004}, and discrete-time, size-structured chemostat models~\cite{MR1821409}.

Nearly fifty years after Volterra's paper, \citet{holt-77} inverted Volterra's model by considering non-competing prey who share a predator. For ordinary differential equation models, \citet{holt-77} showed that the addition of a new prey species to a predator-prey system could reduce the equilibrium density of the original prey species or even drive it extinct. This reduction or exclusion arises as an indirect effect by which  the novel prey increases the predator density and, thereby, increases predation pressure on the original prey species. \citet{holt-77} termed this indirect effect, ``apparent competition'' as to an observer unaware of the shared predator species, the prey appear to be competing. Despite the fundamental ecological importance of this interaction~\citep{holt2017,schreiber_krivan2020}, MathSciNet only list $18$ publications on ``apparent competition'' of which $6$ appear in mathematics journals~\citep{loman-88,jde-04,MR2065415,MR2312293,MR2645916,MR2952099}. All $6$ of these publications use ordinary differential equation models which assume overlapping generations of the prey and predator species. However, some of the most important examples of apparent competition occur in host-parasitoid systems~\citep{holt-93,bonsall-hassell-97,morris2004,holt2017}. 

Due to the tight coupling of their life-cycles, host-parasitoid systems can have discrete, synchronized generations and, consequently, are modeled using difference equations~\citep{hassell-78,mills-getz-96,hassell2000}. As the dynamics of these models can be exceedingly complex, there are few mathematical theorems about their dynamics~\citep[see, however,][]{hsu-03,jamieson_reis2018}. To model apparent competition in host-parasitoid systems, \citet{holt-lawton-93} introduced stochastic difference equations with two, non-competing, host species sharing a common parasitoid. These host species experienced stochastic fluctuations in their intrinsic fitnesses, and the parasitoid species had a stochastic source of immigration. Using a mixture of time-averaging arguments and numerical simulations, \citet{holt-lawton-93} derived a $P^*$-rule: the host species that can support the higher, average parasitoid density excludes the other host species. Regarding their derivation, \citet{holt-lawton-93} wrote ``we have doubtless ignored subtleties in specifying how the parameters must be constrained in their temporal evolution, so that densities are ensured to be bounded away from zero. Numerical simulations suggest that our conclusions hold for reasonable patterns of temporal variability.''

Here, we provide a mathematically rigorous analysis of an extension of \citet{holt-lawton-93}'s model to allow for any number of host species and stochastic variation in the parasitism rates. The analysis includes mathematical proofs of the stabilizing effect of parasitoid immigration, a characterization of  persistence for a single host species and the associated $P^*$ value, and  the $P^*$ rule.  The stochastic, difference-equation model is introduced  in Section~\ref{sec:model}.  The main results about this model are presented in Section~\ref{sec:results}. The results are also illustrated numerically and followed by a discussion of future challenges. To prove the results, we use methods developed by \citet{benaim_schreiber2019} whose key elements are summarized in Section~\ref{sec:BS}. The proofs of the two main theorems for the host-parasitoid models are given in Sections~\ref{sec:proof1} and ~\ref{sec:proof2}.

\section{The Model and Assumptions}\label{sec:model}

We assume that there are $k\ge 1$ host species with densities $x=(x_1,\dots,x_k)$ and one  parasitoid species with density $y$. Let $z(t)=(x(t),y(t))$ be the state of the host-parasitoid community in the $t$-th generation where $x_i(t)\ge 0$ for all $i$ and $y(t)\ge 0$. Each individual of host species $i$ escapes parasitism with probability $\exp(-a_i(t) y(t))$ in the $t$-th generation i.e. the parasitoid attacks are Poisson distributed with mean $a_i(t)y(t)$ on host $i$ where $a_i(t)$ is the attack rate of the parasitoid on host $i$ in the $t$-th generation. Each individual of host species $i$ that escapes parasitism produces $R_i(t)$ offspring that emerge in the next generation. Hosts that do not escape parasitism become parasitoids in the next generation. In addition to this production of parasitoids, there is ``recurrent immigration by the parasitoid from outside the local community''\citep{holt-lawton-93} with $I(t)$ immigrants entering the parasitoid population at the end of the $t$-th generation. Thus, the community dynamics  are
\begin{equation}\label{eq:main}
\begin{aligned}
x_i(t+1)=&R_i(t)x_i(t)\exp(-a_i(t) y(t)) \quad i=1,2,\dots,k\\
y(t+1)=& \sum_{i=1}^k x_i(t)(1-\exp(-a_i(t) y(t)))+I(t).
\end{aligned}
\end{equation}
This model generalizes \citet{holt-lawton-93}'s model by allowing for more than two host species and by allowing the attack rates $a_i(t)$ to stochastically vary. 

To complete the specification of the model, we make the following  assumptions about the $R_i(t)$, $a_i(t)$ and $I(t)$:
\begin{description}
	\item[A1] For each $1\le i\le k$, $R_i(0),R_i(1),R_i(2),\dots$ is a sequence of independent and identically distributed (i.i.d.) random variables taking values in $[\munderbar{R},\bar{R}]$ where $\bar{R}\ge \munderbar{R}>0$.
	\item[A2] For each $1\le i\le k$, $a_i(0),a_i(1),a_i(2),\dots$ is a sequence of i.i.d. random variables taking values in $[\munderbar{a},\bar{a}]$ where $\bar{a}\ge \munderbar{a}>0$.
	\item[A3] $I(0),I(1),I(2),\dots$ is an i.i.d. sequence taking values in $[\munderbar{I},\bar{I}]$ where $\bar{I}\ge \munderbar{I}>0$.
\end{description}
\begin{rem}For several of our main results, the i.i.d. assumption can be relaxed to certain types of stationary sequences (see Remark~\ref{remark} in Section~\ref{sec:BS}). Moreover, recent work by \citet{hening_nguyen2020} allows for relaxing the compactness assumptions.  \end{rem}

\section{Results and Discussion}\label{sec:results}

\begin{figure}
\includegraphics[width=0.5\textwidth]{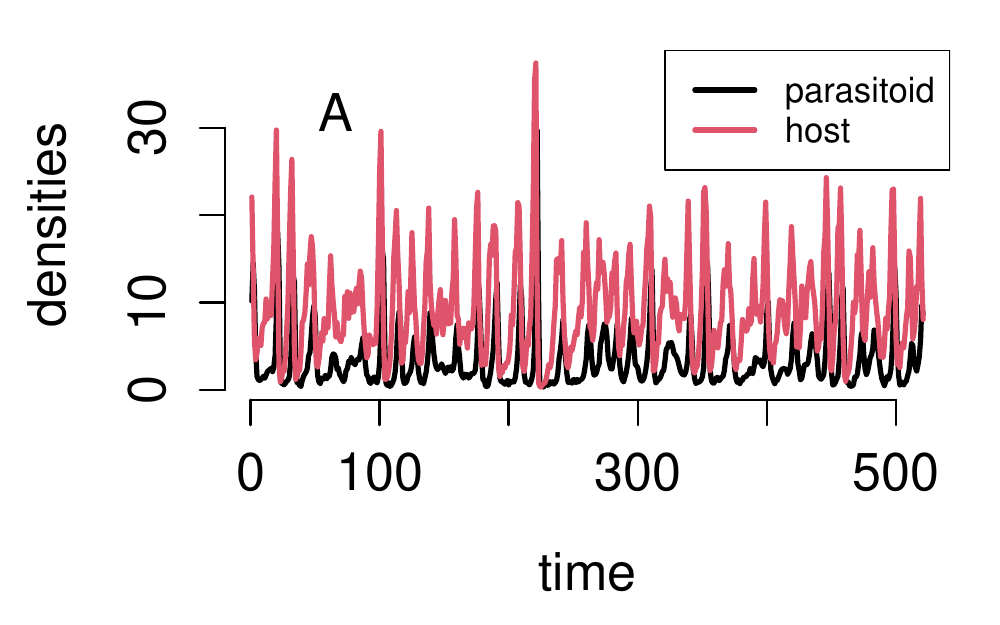}\includegraphics[width=0.5\textwidth]{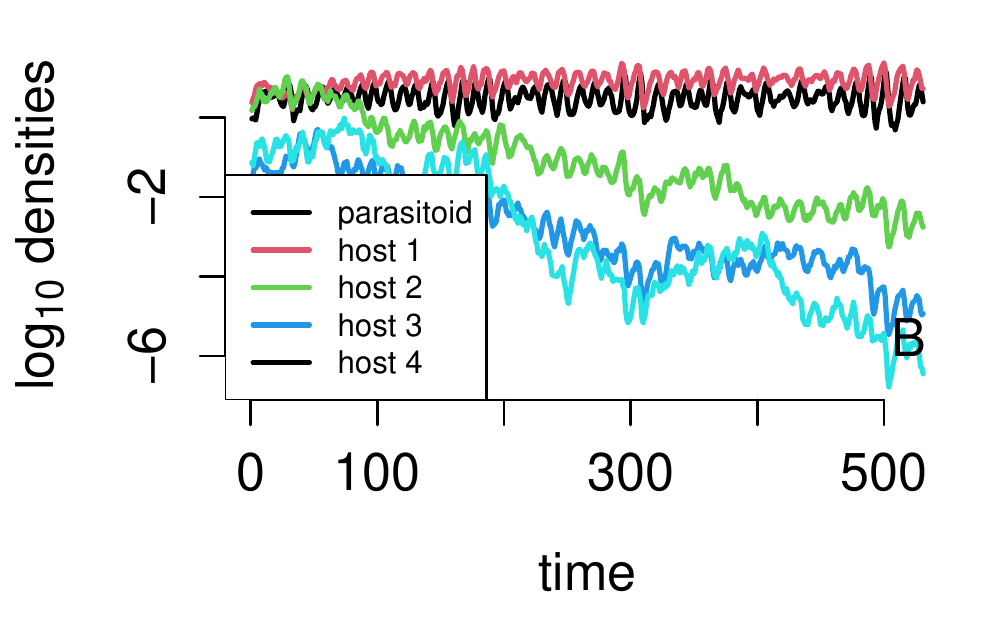}
\caption{Stochastic persistence and the $P^*$ rule for the host-parasitoid model \eqref{eq:main}. In A, there is $k=1$ host species and the condition for stochastic persistence is met. In B, there are $k=4$ host species which only differ in the variance of their $R_i(t)$ terms. Parameter values: $a_i(t)=0.1$ for all $i$ and $t$, $R_i(t)=0.9+1.1\beta^R_i(t)$ where $\beta_i^R(t)$ are $\beta$ distributed with both scale parameters $=k+1-i$, and $I(t)=0.1+0.9\beta^I(t)$ where $\beta^I(t)$ are $\beta$ distributed with both scale parameters $=2$.}\label{fig:one}
\end{figure}

Our first result is to show that solutions of \eqref{eq:main} enter a compact set after a finite amount of time. In contrast, without parasitoid immigration, $k=1$ host species, and constant $R_i$ and $a_i$, equation~\eqref{eq:main} is the Nicholson-Bailey model whose solutions exhibit unbounded oscillations whenever both species are present~\citep{jamieson_reis2018}. The following proposition proves that immigration stabilizes these unbounded oscillations. 

\begin{prop}\label{prop:compact}
There exists a compact set $\s\subset [0,\infty)^k \times [\munderbar{I},\infty)$ such that 
\[
(x(t),y(t))\in \s  \mbox{ for all } t \ge 4
\]
whenever $(x(0),y(0))$ is non-negative i.e. $x_i(0)\ge 0$ for all $i$ and $y(0)\ge 0$.
\end{prop}
  
\begin{proof}
Let $x(0)=(x_1(0),\dots,x_k(0)),y(0)$ be non-negative.  Then 
\[
y(t)\ge \munderbar{I} \mbox{ for all }t\ge 1
\mbox{ and }
 y(t)\ge \sum_{j=1}^k x_j(t-1) (1-\exp(-\munderbar{a} \munderbar{I})) \mbox{ for }t\ge 2.\]
 Define $\alpha= (1-\exp(-\munderbar{a} \munderbar{I})).$ For $t\ge 3$, 
\[
\begin{aligned}
x_i(t)\le &  \bar{R} x_i(t-1) \exp\left(-a_i(t-1) \alpha \sum_{j=1}^k x_j(t-1)\right)\\
\le & \bar{R} x_i(t-1) \exp\left(-\munderbar{a} \alpha x_i(t-1)\right)\\
\le & \frac{\bar{R}}{\munderbar{a} \alpha e} 
\end{aligned}
\]
Thus, for $t\ge 4$,
\[
y(t)\le k\frac{\bar{R}}{\munderbar{a} \alpha e}+\bar{I}.
\]
Setting 
\[
\s=\left[0,\frac{\bar{R}}{\munderbar{a} \alpha e} \right]^k \times 
\left[\munderbar{I},k\frac{\bar{R}}{\munderbar{a} \alpha e}+\bar{I}\right]
\]
completes the proof of the proposition.
\end{proof}

To characterize whether the host persists or not in the presence of the parasitoid, we use two notions of stochastic persistence~\citep[see reviews in][]{jdea-11,schreiber2017}. The first notion corresponds to what \citet{chesson-82} called stochastically bounded coexistence and takes an ensemble point of view. This form of persistence, as shown in equation~\eqref{eq:ensemble} below, implies that probability of a small species density far into the future is small.  The second form of stochastic persistence, introduced in \citep{jmb-11}, takes the perspective of a single, typical realization of the Markov chain. This form of persistence, as shown in equation~\eqref{eq:typical} below, implies that the fraction of time spent below small species densities is small. Figure~\ref{fig:one}A illustrates  the host-parasitoid dynamics in the case of stochastic persistence. For a set $A$, let $\#A$ denote the cardinality of the set.

\begin{thm}\label{thm:one} Assume $k=1$ and assumptions \textbf{A1}--\textbf{A3} hold. If $\E[\ln R_1(t)]<\E[a_1(t)]\E[I(t)]$ and $x_1(0)>0$, then 
\[
\limsup_{t\to\infty} \frac{1}{t}\ln x_1(t)<0 \mbox{ with probability one.}
\]
If $\E[\ln R_1(t)]>\E[a_1(t)]\E[I(t)]$, then there exist $\alpha,\beta>0$ such that for any $\delta>0$
\begin{equation}\label{eq:ensemble}
\limsup_{t\to\infty} \P\left[x_1(t)\le \delta\right]\le \alpha \delta^\beta
\end{equation}
and
\begin{equation}\label{eq:typical}
\limsup_{t\to\infty} \frac{\#\{1\le s\le t: x_1(s)\le \delta\}}{t}\le \alpha \delta^\beta \mbox{ with probability one}
\end{equation}
whenever $x_1(0)>0$ and $y(0)\ge 0$. Moreover, \begin{equation}\label{eq:P*}
\int y \mu(dx,dy)=\E[\ln R_1(t)]/\E[a_1(t)]
\end{equation}
for any invariant measure $\mu(dx,dy)$ supported on $(0,\infty)^2$; the existence of such invariant measures follows from \eqref{eq:typical}.
\end{thm}

Beyond characterizing host persistence, Theorem~\ref{thm:one} via \eqref{eq:P*} provides a mathematical proof of one of \citet{holt-lawton-93}'s conclusions even when the attack rates fluctuate: ``in a fluctuating environment the long-term average parasitoid density [$P^*$] is proportional to the long-term average logarithmic host growth rate.''  The definition of an invariant measure is given in section~\ref{sec:BS}.

Provided that each host species can persist with the parasitoid, our next theorem shows that these long-term average parasitoid densities determine the winner of apparent competition.

\begin{thm}\label{thm:two} Assume assumptions \textbf{A1}--\textbf{A3} hold and
\[
\E[\ln R_i(t)]>\E[a_i(t)]\E[I(t)] \mbox{ for }i=1,2,\dots,k.
\]
Define 
\[
P_i^*:=\E[\ln R_i(t)]/\E[a_i(t)] \mbox{ for }i=1,2,\dots,k.
\]
If $P_1^*>P_i^*$ for $i=2,\dots,k$, then there exist $\alpha,\beta>0$ such that for any $\delta>0$ 
\[
\limsup_{t\to\infty} \frac{1}{t}\ln \max_{2\le i\le k}x_i(t)<0 \mbox{ with probability one}
\]
and
\[
\limsup_{t\to\infty} \frac{\#\{1\le s\le t: x_1(s)\le \delta\}}{t}\le \alpha \delta^\beta \mbox{ with probability one for }\delta\le 1
\]
and
\[
\limsup_{t\to\infty} \P\left[x_1(t)\le \delta\right]\le \alpha \delta^\beta \mbox{ for }\delta\le 1
\]
whenever $\prod_{i=1}^kx_i(0)>0$.
\end{thm}

Thus, this theorem mathematically confirms \citet{holt-lawton-93}'s conclusion: ``regardless of the exact cause of the fluctuations, the outcome should be no different than that expected in a constant environment with stable populations: one host tends to displace alternative hosts from the assemblage, and the winner is the host sustaining the highest average parasitoid density.'' 

\subsection*{Discussion} As noted by \citet{holt-lawton-93},  fluctuations in $R_i(t)$ can influence the winner of apparent competition. For example, suppose that there are two host species with the same mean intrinsic fitness and experiencing the same attack rates i.e. $\E[R_1(t)]=\E[R_2(t)]$ and $a_1(t)=a_2(t)$ for all $t$. However, host $2$ experiences variation in its intrinsic fitness (i.e. $\mbox{Var}[R_2(t)]>0$) while host $1$ experiences no variation (i.e. $\mbox{Var}[R_1(t)]=0$).  As $\ln$ is a concave function, Jensen's inequality implies that $\E[\ln R_1(t)]>\E[\ln R_2(t)]$. Thus, $P_1^*>P_2^*$ and Theorem~\ref{thm:two} implies that host species $2$ is excluded due to having greater variation in its intrinsic fitness. This phenomena is illustrated in Figure~\ref{fig:one}B with $k=4$ host species that only differ in the variances of their intrinsic fitness,  $\mbox{Var}[R_i(t)]$.

Theorems~\ref{thm:one} and \ref{thm:two} are largely possible due to the exponential form of the Poisson escape function (i.e. $\exp(-a_iy)$ due to \citet{nicholson-bailey-35}) and the absence of host-density dependence. In particular, the Poisson escape function assumes that parasitoids are not time-limited and their attacks are randomly distributed among the hosts. Thus, future mathematical challenges include understanding whether or not the $P^*$ rule holds when the escape function accounts for aggregated parasitoid attacks (e.g. the negative binomial form $(1+a_iy/k)^{-k}$ introduced by \citet{may-78}), the escape function accounts for parasitoid time-limitation (e.g. $\exp(-a_iy/(1+b_ix_i))$ as introduced by \citet{hassell1977}), or the hosts experience direct density-dependence (e.g. $R_i(t)\exp(-c_ix_i-a_iy)$ as developed by \citet{may-hassell-anderson-tonkyn-81}). 

Another avenue for future research is to account for feedback between ecology and evolution in the model~\citep{schoener-11}. For deterministic, continuous-time models of two prey sharing a predator, evolution of the predator's attack rate can, by reducing the effects of apparent competition, mediate coexistence but could also lead to oscillatory and chaotic dynamics~\cite{schreiber_burger2011,nrm-15}. Whether similar phenomena arise for the discrete-time, stochastic host-parasitoid model considered here remains to be seen. 

\section{Main Tools from \citet{benaim_schreiber2019}}\label{sec:BS}

To prove Theorems~\ref{thm:one} and \ref{thm:two}, we used methods developed by \citet{benaim_schreiber2019}. These methods apply to models with a mixture of ecological and auxiliary variables (see Remark~\ref{remark} and the proofs in Sections~\ref{sec:proof1},\ref{sec:proof2} for more details). The ecological variables correspond to the densities of $n$ species given by $u=(u_1,u_2,\dots,u_n)\in [0,\infty)^n=:\R^n_+$. The species dynamics interact with the auxiliary variable $v$ which lies in $(-\infty,\infty)^m=:\R^m$. In the proof of Theorem~\ref{thm:one} the parasitoid density is treated as an auxiliary variable, while in the proof of Theorem~\ref{thm:two} the densities of host species $2$ through $k$ also are used as auxiliary variables.

The ecological and auxiliary variables may be influenced by stochastic forces captured by a sequence of independent and identically distributed (i.i.d.) random variables $\xi(1),\xi(2),\dots$ taking values in a Polish space $\Xi$ i.e. a separable completely metrizable topological space. The stochastic difference equations considered by \citet{benaim_schreiber2019} are of the form:
\begin{equation}\label{eq:main2}
\begin{aligned}
u_i(t+1) &= u_i(t) f_i(u(t),v(t),\xi(t)) \quad i=1,2,\dots,n& \mbox{ (species densities)}\\
v(t+1)&= G(u(t),v(t),\xi(t))& \mbox{ (auxiliary variables)}.
\end{aligned}
\end{equation}
with standing assumptions:
\begin{description}
\item [B1] For each $i=1,2,\dots,n$, the fitness function $f_i(z,\xi)$ is continuous in $z=(u,v)$, measurable in $(z,\xi)$, and strictly positive.
\item [B2] The auxiliary variable update function $G$ is continuous in $z=(u,v)$ and measurable in $(z,\xi)$.
\item [B3] There is a compact subset $\s$ of $\R^n_+\times \R^m$ such that all solutions $z(t)=(u(t),v(t))$ to \eqref{eq:main} satisfy $z(t)\in \s$ for $t$ sufficiently large. 
\item [B4] For all $i=1,2,\dots,n$, $\sup_{z,\xi}|\log f_i(z,\xi)|<\infty$.
\end{description}
Beyond \eqref{eq:main}, many finite-dimensional, discrete-time population models satisfy assumptions \textbf{B1--B4} \citep[see,e.g.,][]{benaim_schreiber2019,schreiber2020}.

\begin{rem}\label{remark}
For the proofs of Theorems~\ref{thm:one} and \ref{thm:two}, host species $1$ of \eqref{eq:main} is always treated as a species density (e.g. $u_1=x_1$), the parasitoid density is always treated as an auxiliary variable (e.g $v_1=y$), the host species $2$ through $k$ are treated either as species densities (e.g. $u_i=x_i$ for $i=2,\dots,k$) or as auxiliary variables (e.g. $v_i=x_i$ for $i=2,\dots,k$), and the i.i.d. random variables $\xi(t)$ equal $(R_i(t),a_i(t),I(t))$. Our assumptions \textbf{A1}--\textbf{A3} and Proposition~\ref{prop:compact} ensure that assumptions \textbf{B1}-\textbf{B4} hold for \eqref{eq:main}.\end{rem}

\begin{rem}\label{remark} One can also account for temporal correlations in $R_i(t)$, $a_i(t)$, and $I(t)$ using additional auxiliary variables. For example, $v_i(t)=\ln R_i(t)$ could be modeled as a first-order autoregressive process given by $v_i(t+1)=\alpha_i v_i(t)+\xi_i(t+1)$ where the $\xi_i(t)$ are i.i.d. Provided that $|\alpha_i|<\infty$ and $\xi_i(t)$ take values in a compact set, the assumptions \textbf{B1}--\textbf{B4} hold. Alternatively, one can model the fluctuations in $R_i(t)$ using a finite-state Markov chain. To see how, suppose that $R_i(t)$ takes on a finite number of distinct,positive values, $r_1,\dots,r_\ell$, with transition probabilities $p_{ij}$ i.e. $\P[R_i(t+1)=r_j|R_i(t)=r_i]=p_{ij}$. One can represent this Markov chain $v_i(t)=R_i(t)$ as a composition of random maps by defining  $\xi(t)=(\xi_1(t),\dots,\xi_\ell(t))$ to be a random vector such that $\P[\xi_i(t)=r_j]=p_{ij}$, and defining $G_i(v_i,\xi)=\xi_{\pi(v_i)}$ where $\pi(r_i)=i$ for $1\le i\le \ell.$ Proposition~\ref{prop:compact} holds when temporal correlations in the $R_i(t),a_i(t),I(t)$ are modelled in this way. The first two conclusions about exclusion and stochastic persistence of Theorem~\ref{thm:one} holds when  $\E[\ln R_1(t)]<\E[a_1(t)I(t)]$ and $\E[\ln R_1(t)]>\E[a_1(t)I(t)]$, respectively. However, equation~\eqref{eq:P*} of Theorem~\ref{thm:one} need not hold if the attack rates $a_i(t)$ exhibit temporal autocorrelations and,  consequently, Theorem~\ref{thm:exclusion} need not hold in this case. \end{rem}

To evaluate whether species are increasing or decreasing when rare, we consider their per-capita growth rate averaged over the fluctuations in  $u(t)$, $v(t)$, and  $\xi(t+1)$. To this end, recall that a Borel probability measure $\mu$ on $\s$ is \emph{an invariant probability measure} if for all continuous functions $h:\s\to\R$
\[
\int_\s h(z)\mu(dz)= \int_\s \E[h(Z(1))|Z(0)=z]\mu(dz).
\]
An invariant probability measure $\mu$ is \emph{an ergodic probability measure} if it can not be written as a non-trivial convex combination of invariant probability measures. For any invariant probability measure $\mu$, define $r_i(\mu)$ as  the \emph{realized per-capita growth rate} of population $i$:
\begin{equation}\label{eq:per-capita}
r_i(\mu)=\int_\s \E[\log f_i(z,\xi(t)) ] \mu(dz).
\end{equation}

For any ergodic probability measure $\mu$, define \emph{the species supported by $\mu$}, denoted $S(\mu)$, to be the unique subset $I\subset \{1,2,\dots,n\}$ such that $\mu(\{(u,v)\in\s:u_i>0$ iff $i\in I\})=1$. The following proposition implies that $r_i(\mu)=0$ for all $i\in S(\mu)$. Alternatively, for $i\notin S(\mu)$, $r_i(\mu)$ need not be zero in which case $r_i(\mu)$ measures the rate of growth of species $i$ when introduced at infinitesimally small densities. For $i\notin S(\mu)$, $r_i(\mu)$ is also known as the \emph{external Lyapunov exponent} of $\mu$. The following result is proven in \cite[Proposition 1]{benaim_schreiber2019}.

\begin{prop}
\label{prop:zero}  Let $\mu$ be an ergodic  probability measure. Then $r_i(\mu)=0$ for all $i\in S(\mu).$
\end{prop}

Following the approach introduced by Josef Hofbauer~\citep{hofbauer-81,hofbauer_sigmund1998}, the following Theorem from \citep[Theorem 1]{benaim_schreiber2019} gives a sufficient condition for stochastic persistence. We make use of this theorem for the proofs of both Theorems~\ref{thm:one} and \ref{thm:two}. To state this theorem, define the extinction set as 
\[
\s_0=\{(u,v)\in \s: \prod_{i=1}^n u_i=0\}.
\]

\begin{thm}\label{thm:persist} If 
\begin{equation}\label{eq:criterion}
\mbox{there exist positive }p_1,\dots,p_n \mbox{ s.t. } \sum_i p_i r_i(\mu)>0\mbox{ for all ergodic $\mu$ with $\mu(\s_0)=1$}.
\end{equation}
 holds,
then there exist  $a,b>0$ such that for all $\delta\le 1$ and $Z(0)=z\in \s\setminus \s_0$
\[
\mbox{\rm(persistence in probability) }\limsup_{t\to\infty}\P\left[\min_{1\le i\le n}u_i(t)\le \delta \right]\le a (\delta)^b 
\]
and
\[
\mbox{\rm(almost-sure persistence) }\limsup_{t\to\infty}\frac{\#\{1\le s\le t: \min_{1\le i\le n}u_i(s)\le \delta\}}{t}\le a (\delta)^b 
\mbox{ almost surely.}\]
\end{thm}

To identify when species are driven extinct, we consider the case when there is a subset $I\subset\{1,2,\dots,n\}$ of species that can not be invaded. Define
\[
\s^{I}:=\{(u,v)\in\s|u_j= 0\mbox{ whenever }j\notin I\}
\]
and for $\delta> 0$, define \[\s^{I,\delta}:=\{(u,v)\in\s|u_j\le \delta\mbox{ whenever }j\notin I\}.\]  We say \emph{$\s^I$ is accessible} if for all $\delta>0$, there exists $\gamma>0$ such that
\[
\P[ z(t)\in \s^{I,\delta} \mbox{ for some }t\ge 1]\ge \gamma
\]
whenever $Z(0)=(u,v)$ satisfies $\prod_{i}u_i>0.$ Intuitively, this accessibility conditions states that with probability one, the process will eventually enter any neighborhood of $\s^I$. As the process is Markov, this implies that the process will enter this neighborhood infinitely often. The following Theorem follows from \citep[Thm. 3]{benaim_schreiber2019}.

\begin{thm}\label{thm:exclusion} Let $I$ be a strict subset of $\{1,2,\dots,n\}$. Assume
\begin{enumerate}
\item[(i)]\eqref{eq:main} restricted $\s^I$ satisfies that there exist $p_i>0$ for $i\in I$ and $\sum_{i\in I}p_ir_i(\mu)>0$ for ergodic $\mu$ with $\mu(\s^I_0)=1$ where $\s^I_0:=\{z=(u,v)\in \s^I: \prod_{i\in I}u_i=0\}$, 
\item[(ii)] $r_j(\mu)<0$ for any $j\notin I$ and ergodic $\mu$ satisfying $S(\mu)=I$, and
\item[(iii)] $\s^{I}$ is accessible.
\end{enumerate}
Then
\begin{equation}\label{eq:exclude}
\P\left[\limsup_{t\to\infty} \frac{1}{t}\log\rm{dist}(z(t),\s_0)<0\right]=1 \mbox{ whenever }Z(0)=z\in\s.
\end{equation}
\end{thm}
Condition (i) in Theorem~\ref{thm:exclusion} ensures the set of species in $I$ coexist in the sense of stochastic persistence. Condition (ii) implies that the per-capita growth rates are negative for all of the species not in $I$. Conditions (i) and (ii) are sufficient to ensure the local attractivity of $\s^I$ in a stochastic sense--see Theorem 2 in \citep{benaim_schreiber2019}. Condition (iii) ensures the global attractivity with probability one. 

\section{Proof of Theorem~\ref{thm:one}}\label{sec:proof1}
Assume that $\alpha:=\E[\ln R_1(t)]- \E[a_1(t)]\E[I(t)]<0$.  As $y(t)\ge I(t-1)$ for all $t\ge 1$, it follows that for all $t\ge 1$
\[
\begin{aligned}
x_1(t+1)\le& R_1(t)x_1(t)\exp(-a_1(t)I(t-1))\\
\le  &x_1(0)\bar{R}\prod_{s=1}^{t}R_1(s)\exp(-a_1(s)I(s-1)).
\end{aligned}
\]
The strong law of large numbers implies that with probability one 
\[
\limsup_{t\to\infty}\frac{\ln x_1(t+1)}{t+1}\le 
\lim_{t\to\infty} \frac{1}{t+1}\left(\ln \bar{R} x_1(0)+\sum_{s=1}^t
\left( \ln R_1(s)- a_1(s)I(s-1)\right)\right)=\alpha<0. 
\]

Now assume $\alpha:=\E[\ln R_1(t)]-\E[a_1(t)]\E[I(t)]>0$. We will use Theorem~\ref{thm:persist} with $n=1$, $u_1=x_1$ and $v_1=y$ in \eqref{eq:main2}. On $\s_0=\{(x_1,y)\in \s: x_1=0\}$, the dynamics are given by $x_1(t)=0$ for all $t$ and $y(t)=I(t)$ for all $t\ge 0$. As the $I(t)$ are i.i.d., the only ergodic invariant measure $\mu(dx,dy)$ for the dynamics on $\s_0$ is determined by the law $m(dy)$ of $I(t)$ i.e. $\int_{\{0\}\times A}\mu(dx,dy)=\int_A m(dy)$ for any Borel set $A\subset [0,\infty)$.  For this invariant measure, the per-capita growth rate of the host equals
\[
r_1(\mu)=\int_0^\infty \E[\ln R_1(t)] - a_1(t)y]m(dy)=
\alpha>0.
\]
Hence, Theorem~\ref{thm:persist} implies the first two conclusions for the case of $\alpha>0$. For the final conclusion, let $\nu(dx_1,dy)$ be any ergodic measure such that $\nu(\s\setminus\s_0)=1$. Then, Proposition~\ref{prop:zero} implies 
\begin{equation}\label{eq:ergodic}
0= r_1(\nu)=\int \E[\ln R_1(t)-a_1(t)y]\nu(dx_1,dy)=\E[\ln R_1(t)]-\E[a_1(t)]\int y \nu(dx_1,dy).
\end{equation}
By the ergodic decomposition theorem~\citep[Theorem 4.1.12]{katok-hasselblatt-95}, every invariant probability measure $\mu$ satisfying $\mu(\s\setminus \s_0)=1$ is a convex combination of ergodic measures $\nu$ satisfying $\nu(\s\setminus \s_0)=1$. \eqref{eq:ergodic} applied to each of these ergodic measures $\nu$ in the decomposition of $\mu$ implies the final conclusion of the case $\alpha>0.$

\section{Proof of Theorem~\ref{thm:two}}\label{sec:proof2}

First, we show that host species $1$ is stochastically persistent. To this end, we use Theorem~\ref{thm:persist} with $u=x_1$ and $v=(x_2,\dots,x_k,y)$ in \eqref{eq:main2} i.e. the other host species and the parasitoid are treated as auxiliary variables. For these choices, the extinction set is $\s_0=\{(x,y)\in \s: x_1=0\}.$ Let $\mu(dx,dy)$ be an ergodic invariant probability measure on $\s_0$. Then either $\mu(dx,dy)$ supports no host species in which case $r_1(\mu)=\E[\ln R_1(t)]-\E[a_1(t)]\E[I(t)]>0$ or $\mu(dx,dy)$ supports at least one host species $i\ge 2$. In the latter case, Proposition~\ref{prop:zero} implies that 
\[
0=r_i(\mu)=\E[\ln R_i(t)]-a_i(t) \int y \mu(dx,dy)
\]
and therefore $\int y\mu(dx,dy)=\E[\ln R_i(t)]/\E[a_i(t)]$. On the other hand, 
\[
r_1(\mu)=\E[\ln R_1(t)]-\E[a_1(t)] \int y\mu(dx,dy)=\E[\ln R_1(t)]-\E[a_1(t)]\frac{\E[\ln R_i(t)]}{\E[a_i(t)]}
\]
As $P_1^*=\E[\ln R_1(t)]/\E[a_1(t)]> \E[\ln R_i(t)]/\E[a_i(t)]=P_i^*$, it follows that $r_1(\mu)>0$. As we have shown that $r_1(\mu)>0$ for all ergodic measures supported by $\s_0$, Theorem~\ref{thm:persist} with  $p_1=1$ implies stochastic persistence as claimed.

Next, we show that for $i\ge 2$
\[
\limsup_{t\to\infty}\frac{\ln x_i(t)}{t}<0 \mbox{ with probability one}
\]
whenever $x_i(0)x_1(0)>0$. To prove this conclusion, we verify the conditions of Theorem~\ref{thm:exclusion} with $u=(x_1,\dots,x_k)$ and $v=y$ (i.e. only the parasitoid is an auxiliary variable) in \eqref{eq:main2}, and $I=\{1\}$ in conditions (i)--(iii) in Theorem~\ref{thm:exclusion}. For these choices, $\s^I=\{(x,y)\in \s: x_2=\dots=x_k=0\}$. Theorem~\ref{thm:one} applied to the $x_1-y$ subsystem implies condition (i) of Theorem~\ref{thm:exclusion}. Next, we verify condition (ii) i.e. $r_i(\mu)<0$ for all $2\le i\le n$ and ergodic probability measures $\mu$ such that $S(\mu)=\{1\}=I$. Let $\mu$ be such an ergodic measure. Theorem~\ref{thm:one} implies that  for $i\ge 2$
\[
r_i(\mu)=\E[\ln R_i(t)]-\E[a_i(t)] \frac{\E[\ln R_1(t)]}{\E[a_1(t)]}=\E[a_i(t)](P_i^*-P_1^*)<0
\]
as we have assumed that $P_1^*>P_i^*.$
Next we verify assumption (iii) of Theorem~\ref{thm:exclusion}. Consider any initial condition $(x(0),y(0))\in \s$ such that $x_1(0)>0$ and $\sum_{i=2}^k x_i(0)>0$. Define the occupational measure
\[
\Pi_t=\frac{1}{t}\sum_{s=1}^t \delta_{z(s)}
\] 
where $\delta_z$ is a Dirac measure at $z$ i.e. for any Borel set $A\subset \s$, $\delta_z(A)=1$ if $z\in A$ and $0$ otherwise. By Lemma 4 of \citep{benaim_schreiber2019} and the stochastic persistence of host species $1$ from the first part of this proof,  the weak* limit points $\mu$ of $\Pi_t$ as $t\to\infty$ are, with probability one, invariant probability measures that satisfy $\mu(\{(x,y)\in \s: x_1=0\})=0$ and $r_1(\mu)=0$. Hence, for these weak* limit points $\mu$, we have $\int y \mu(dx,dy)=P_1^*$. For such a $\mu$, we claim that $\mu(\{(x,y):x_1>0,x_2=\dots=x_n=0\})=1$. Suppose, to the contrary, that for some $i\ge 2$, $\mu(\{(x,y)\in\s:x_1>0,x_i>0\})>0$. Then, by the  ergodic decomposition theorem~\citep[Theorem 4.1.12]{katok-hasselblatt-95}, there is an ergodic probability measure $\nu$ such that  
$\nu(\{(x,y)\in\s:x_1>0,x_i>0\})=1$. By Proposition~\ref{prop:zero},  \[0=r_i(\nu)/\E[a_i(t)]=\E[\ln R_i(t)]/\E[a_i(t)] - \E[\ln R_1(t)]/\E[a_1(t)]=P_i^*-P_1^*,\] a contradiction to our assumption that $P_1^*>P_i^*$. Hence, with probability one, the weak* limit points $\mu$ of $\Pi_t$ as $t\to\infty$ satisfy $\mu(\{(x,y)\in\s:x_1>0,x_2=\dots=x_n=0\})=1$ as claimed. In particular, this implies for any neighborhood $U$ of $\{(x,y)\in\s:x_1>0,x_2=\dots=x_n=0\}$, $z(t)$ enters $U$ infinitely often with probability one. Hence, condition (iii) of Theorem~\ref{thm:exclusion} is satisfied and \eqref{eq:exclude} implies 
\[
\limsup_{t\to\infty}\frac{1}{t}\ln \max_{2\le i \le k}x_i(t) <0 \mbox{ with probability one}\]
as claimed. $\qed$

\bibliographystyle{unsrtnat}
\bibliography{seb}

\end{document}